\documentclass[conference]{IEEEtran}

\usepackage{mathtools}
\usepackage{amssymb}
\usepackage{amsmath}
\usepackage{algorithm}
\usepackage{enumerate,multirow}
\usepackage{algpseudocode}

\usepackage{siunitx}
\usepackage{multirow}
\usepackage{multicol}
\usepackage{color}
\usepackage{xcolor}

\usepackage{comment}

\usepackage{geometry}
\usepackage{graphicx}
\usepackage{graphicx}
\usepackage{tikz}
\usetikzlibrary{trees}
\date{}

\usepackage{calc}
\newcolumntype{M}[1]{>{\centering\arraybackslash}m{#1}}
\newcolumntype{N}{@{}m{0pt}@{}}

\usepackage{mdwmath}
\usepackage{blindtext}
\usepackage{eqparbox}
\usepackage{stfloats}
\usepackage[numbers,sort&compress]{natbib}
\usepackage{algorithm}
\usepackage{algpseudocode}
\IEEEoverridecommandlockouts
\title{\Huge
{Bounds and New Constructions for Girth-Constrained Regular Bipartite Graphs}}
\author{\IEEEauthorblockN{Sheida Rabeti, Mohsen Moradi, and Hessam Mahdavifar} 
\IEEEauthorblockA{Department of Electrical and Computer Engineering, Northeastern University, Boston, MA 02115, USA \\ 
Email: \{rabeti.s, m.moradi, h.mahdavifar\}@northeastern.edu}
\thanks{This work was supported by NSF under Grant CCF-2415440 and the Center for Ubiquitous Connectivity (CUbiC) under the JUMP 2.0 program.}

}

\newtheorem{theorem}{{Theorem}}

\newtheorem{proposition}[theorem]{{Proposition}}

\newtheorem{definition}{{Definition}}

\newcommand{\cC}{{\cal C}} 
\newcommand{\cD}{{\cal D}}

\newcommand{\cG}{{\cal G}}

\newcommand{\cO}{{\cal O}}

\newcommand{\cV}{{\cal V}}

\DeclareMathAlphabet{\mathbfsl}{OT1}{ppl}{b}{it} 

\newcommand{\be}[1]{\begin{equation}\label{#1}}
\newcommand{\ee}{\end{equation}}

\renewcommand{\leq}{\leqslant}
 
\renewcommand{\geq}{\geqslant}

\renewcommand{\Bbb}{\mathbb}
 
\newcommand{\N}{{\Bbb N}}

\newcommand{\Cref}[1]{Co\-ro\-lla\-ry\,\ref{#1}}

\usepackage{hvlogos}
\usepackage{qtree}
\usepackage{tikz}
\usepackage{multirow}
\usepackage{hhline}
\usepackage{hyperref}

\begin{document}

\vspace{10mm}
\maketitle

\begin{abstract}
In this paper, we explore the design and analysis of regular bipartite graphs motivated by their application in low-density parity-check (LDPC) codes specifically with constrained girth and in the high-rate regime. We focus on the relation between the girth of the graph, and the size of the sets of variable and check nodes. We derive bounds on the size of the vertices in regular bipartite graphs, showing how the required number of check nodes grows with respect to the number of variable nodes as girth grows large. Furthermore, we present two constructions for bipartite graphs with girth $\cG = 8$; one based on a greedy construction of $(w_c, w_r)$-regular graphs, and another based on semi-regular graphs which have uniform column weight distribution with a sublinear number of check nodes. The second construction leverages sequences of integers without any length-$3$ arithmetic progression and is asymptotically optimal while maintaining a girth of $8$. Also, both constructions can offer sparse parity-check matrices for high-rate codes with medium-to-large block lengths. Our results solely focus on the graph-theoretic problem but can potentially contribute to the ongoing effort to design LDPC codes with high girth and minimum distance, specifically in high code rates.
\end{abstract}

\section{Introduction}
Graphs have found extensive applications across various fields, particularly in communication theory, where they play a significant role in the design and analysis of error-correcting codes. One notable class of graph-based codes that has attracted significant attention is low-density parity-check (LDPC) codes. These codes, which are represented by sparse bipartite graphs, offer exceptional error correction capabilities and have been widely adopted in modern communication systems, especially in wireless and cellular communications systems, due to their near-optimal performance and efficient decoding algorithms \cite{mackay1999good, richardson2001capacity}. In general, LDPC codes, as introduced by Gallager in \cite{gallager1962low}, are defined by a sparse structure in the parity-check matrix. Specifically, this structure often involves assigning a small, fixed number of ones to each row and column. The girth of a graph is the length of the smallest cycle that appears in the graph. Throughout this paper, we denote the girth by $\cG.$ It is well-known that parity-check matrices with small minimum distances and/or small girths adversely affect the performance of the belief-propagation (BP) decoders. In fact, these parameters are highly correlated with the size of trapping sets for the BP decoder. In addition, short loops imply that the messages incoming to the variable/check nodes will be conditionally dependent after only a few decoding iterations, which may serve as another factor that affects the performance negatively \cite{price2017survey}. Consequently, designing bipartite graphs with relatively high girths has been a major focus in the literature of LDPC codes; see, e.g., \cite{smarandache2022unifying,smarandache2012quasi,kim2007quasi,tasdighi2016efficient} and references therein.
Interestingly, this problem is also a classical problem that has been of interest in graph theory and combinatorics; see, e.g., \cite{tutte1947family, alon2002moore,erdos1964extremal,de1991maximum}. 

Motivated by the above, in this paper, we attempt to provide bounds specifically for bipartite graphs, arguing that not only there exist conditions on the overall size of vertices when the girth is constrained, but also there exist certain constraints on the relative size of each part of the bipartite graph to each other. A regular bipartite graph (or biregular bipartite graph, as often referred to in the literature on graph theory when the underlying graph is bipartite) is defined as a graph with two parts $\cV_c, \cV_r$ such that the degree of each vertex $v_c \in \cV_c$ is $w_c$, and the degree of each vertex $v_r \in \cV_r$ is $w_r$. For simplicity, we call such graphs $(w_c, w_r)$-regular graphs. Note that, contrary to the mainstream literature on LDPC codes, we do not assume $w_c$ and/or $w_r$ to be fixed or regarded as constants, unless specifically mentioned. However, we keep the same notation for the sake of specifying whether the graph is regular or \textit{semi-regular}, where semi-regular refers to regularity only in the column weight distribution. Assuming $w_c \leq w_r$, we show that for $\cG = 8$, $|\cV_r| \geq \cO(|\cV_c|^{1/2})$, for $\cG = 10, 12$, $|\cV_r| \geq \cO(|\cV_c|^{2/3})$, and for $\cG = 14, 16$ we have $|\cV_r| \geq \cO(|\cV_c|^{3/4})$ while characterizing the multiplicative constant in these results. In particular, by solving the corresponding quadratical, qubic, and quartic equations, we can find some tighter bounds compared with existing results, i.e., \cite{de1991maximum, lam2001graphs}. 
Moreover, in the recent work \cite{mahdavifar2024high}, fair-density parity-check (FDPC) codes were introduced that target high-rate applications, while demonstrating excellent performance. These codes operate in the regime with the number of parity-check equations $m = |\cV_r|$ being sublinear with the block-length $n = |\cV_c|$, which serves as another motivation to study lower bounds on $|\cV_r|$, although sublinear in $|\cV_c|$, given a desired girth. 

Furthermore, we provide two constructions for the specific case of $\cG = 8$. The first construction is a $(w_c, w_r)$-regular graph constructed with a greedy-type method and results in a code with minimum distance $d_{min} = 2^{w_c}$ when used as the Tanner graph of the code. The second construction is semi-regular and the size of the parity-check nodes in this construction is asymptotically $\cO(\sqrt n)$, while providing girth $\cG = 8$. We will observe that this is nearly optimal compared to the lower bound we provide. 

The rest of this paper is organized as follows. In Section \ref{sec:bound}, we provide bounds on the size of the vertices for regular bipartite graphs. In Section \ref{sec:reg}, we construct a regular bipartite graph with a notably large minimum distance, and in Section \ref{sec:semi}, we propose a semi-regular graph which has a smaller vertex size compared to the regular construction. Finally, we conclude the paper in Section \ref{sec:conclusion}.

\section{Bounds On The Size of Regular Bipartite Graphs}
\label{sec:bound}
In this section, we derive bounds on the size of $(w_c, w_r)$-regular bipartite graphs up to the girth $16$. The main idea follows a well-known approach based on the fact that for any graph with girth $\cG = 2\ell$, and for every vertex $v$, there exists a tree of height $\ell$ rooted at $v$ (see, e.g., \cite{Biggs1998ConstructionsFC}). Let $m = |\cV_r|$, and $n = |\cV_c|$ denote the number of vertices with degrees $w_r$, and $w_c$ respectively. In the following, we provide a lower bound on $m$ in terms of $n$ and $w_c$. 

\subsection{Girth $8$}
\begin{center}
    \begin{tikzpicture}[
    scale=.44, 
    every node/.style={font=\normalsize}, 
    circ/.style={circle, fill=black, inner sep=0pt, minimum size=7pt},
    sq/.style={rectangle, draw, inner sep=0pt, minimum size=7pt}
]

        \node[sq,pin=90:{Check nodes}] (n11) at (0,0) {};

        \node[circ,pin=90:{Variable nodes}] (n21) at (-2.5,-1) {};
        \node[circ] (n22) at (2.5,-1) {};

        \node[sq] (n31) at (-4.5,-4) {};
        \node[sq] (n32) at (-0.5,-4) {};
        \node[sq] (n33) at (0.5,-4) {};
        \node[sq] (n34) at (4.5,-4) {};

        \node[circ] (n41) at (1,-7) {};
        \node[circ] (n42) at (-2,-7) {};
        \node[circ] (n43) at (3,-7) {};
        \node[circ] (n44) at (6,-7) {};
        
        \node[] (n51) at (1.5,-8.5) {};
        \node[] (n52) at (3.3,-8.5) {};
        \node[] (n53) at (5.5,-8.5) {};

        \node[] (n71) at (0,-9) {};
        \node[] (n72) at (6,-9) {};
        \node[] (n73) at (-4.5,-9) {};
        
        \node[sq] (n61) at (0,-11) {};
            \begin{scope}[every node/.style={above,sloped,font=\LARGE,scale=.6}]

        \draw (n11) -- node {$1$} (n21);
        \draw (n11) -- node {$w_r$} (n22);

        \draw (n21) -- node {$1$} (n31);
        \draw (n21) -- node {$w_c-1$} (n32);
        \draw (n22) -- node {$1$} (n33);
        \draw (n22) -- node {$w_c-1$} (n34);

        \draw (n32) -- node {$w_r-1$} (n41);
        \draw (n32) -- node {$1$} (n42);
        \draw (n34) -- node {$1$} (n43);
        \draw (n34) -- node {$w_r-1$} (n44); 

        \draw (n43) -- node {$1$} (n51);
        \draw (n43) -- node {$i$} (n52);
        \draw (n43) -- node {$w_c-1$} (n53);

        \draw (n61) -- node {$j$} (n71);
        \draw (n61) -- node {$w_r$} (n72);
        \draw (n61) -- node {$1$} (n73);
        \end{scope}
    \end{tikzpicture}
\end{center}
\begin{theorem}
\label{thm-8bnd}
    Let $w_c \in \N$. For every regular bipartite graph of column-weight $w_c$ and with girth $\cG = 8$, the size of the smaller part $m$ satisfies:
    \begin{equation}
        m \geq \frac{-w_c(w_c-2)+w_c\sqrt{(w_c-2)^2+4(w_c-1)n}}{2}.
    \end{equation}
\end{theorem}
\begin{proof}
    Using a $3$-level tree rooted at an arbitrary check node, we can determine that the total number of check nodes in the tree, excluding the last layer, is equal to
    \begin{equation}
        m' = 1 + w_r(w_c-1).
    \end{equation}
    By double-counting the total number of edges in the last layer, we find that
    \begin{equation}
        w_r(w_r-1)(w_c-1)^2 \leq w_r. m_L
    \end{equation}
    where $m_L$ denotes the number of check nodes in the last layer of the tree.
    
    Note that the equality may not hold because connecting all check nodes in the last layer may result in some $4$-cycles or $6$-cycles. Thus,
    \begin{equation}
        m = m^{'} + m_L \geq 1 + w_r(w_c-1) + (w_r-1)(w_c-1)^2.
    \end{equation}
    Substituting the $w_r = \frac{w_cn}{m}$ into the expression yields
    \begin{equation}
    \label{eq-boun-8}
        m^2 + w_c(w_c - 2) m - w_c^2(w_c-1)n \geq 0.
    \end{equation}
    Solving the equation \eqref{eq-boun-8} gives the positive root of
    \begin{equation}
        x = \frac{-w_c(w_c-2) + w_c\sqrt{(w_c-2)^2+4(w_c-1)n}}{2}.
    \end{equation}
    By the convexity of the quadratic equation and since the other root is negative, we find that $m \geq x$.
\end{proof}

The relevance of the setup in Theorem \ref{thm-8bnd} to channel coding can be explained as follows. For example, consider the FDPC codes of \cite{mahdavifar2024high} that operate in high rates. The high-rate regime of this code inherently results in the row weight of the parity-check matrix being much larger than the column weight. Hence, it is relevant to ask the following question for LDPC (or FDPC) coding at high rate: Given a certain constraint on the column weight and the girth of the Tanner graph, what is the minimum required redundancy at a given block length? That is what Theorem \ref{thm-8bnd} provides a lower bound for. 

\noindent{Remark $1$:}
Note that in \cite{mahdavifar2024high}, the tanner graph corresponding to the base parity-check matrix has a girth of $8$ with parameters $w_c = 2$ and $m = 2\sqrt{n}$. Substituting these parameters into equation \eqref{thm-8bnd} yields equality, demonstrating the optimality of the base matrix of \cite{mahdavifar2024high} in the case of $w_c=2$.
\subsection{$10 \leq \cG \leq 16$}
Similarly, at Step $1$, we determine the bounds corresponding to girths $10$, $12$, $14$, and $16$ using trees rooted at a check node with heights $5$, $6$, $7$, and $8$, respectively. At Step $2$, further results are obtained by solving cubic and quartic equations. Due to the complexity of the roots, we focus on the lower bounds in an order-wise manner. However, the specific constants can be essentially derived leveraging formulas for roots of polynomials up to degree four though in rather tedious forms. It is known by Abel's impossibility theorem that such formulas do not exist for higher degree polynomials, and that is the main reason we stop at $\cG = 16$.

\subsubsection{Step $1$}
The inequalities are found as below:
\begin{center}
\begin{tabular}{ |p{1.3cm}||p{6.5cm}|  }
 \hline
 \multicolumn{2}{|c|}{Inequalities} \\
 \hhline{|=|=|}
 Girth $10$ & $m \geq 1 + w_r(w_c-1) + w_r(w_r-1)(w_c-1)^2$ \\
 \hline
 Girth $12$ & $m \geq 1 + w_r(w_c-1) + w_r(w_r-1)(w_c-1)^2 + (w_r-1)^2(w_c-1)^3$\\
 \hline
 Girth $14$ & $m \geq 1 + w_r(w_c-1) + w_r(w_r-1)(w_c-1)^2 + w_r(w_r-1)^2(w_c-1)^3$  \\
 \hline
 Girth $16$ & $m \geq 1 + w_r(w_c-1) + w_r(w_r-1)(w_c-1)^2 + w_r(w_r-1)^2(w_c-1)^3 + (w_r-1)^3(w_c-1)^4$ \\
 \hline
\end{tabular}
\end{center}
\subsubsection{Step $2$} 
Substituting $w_r = \frac{w_c n}{m}$ yields polynomials $P_{\cG}(m)$, corresponding to each girth $\cG$, that satisfy $P_{\cG}(m) \geq 0$. For $a_{i,j} \geq 0$, we have
\begin{center}
\begin{tabular}{ |p{1.3cm}||p{6.5cm}|  }
 \hline
 \multicolumn{2}{|c|}{Polynomials} \\
 \hhline{|=|=|}
 $P_{10}(m)$ & $m^3 - m^2 + a_{1,1} n m - a_{1,2}n^2$ \\
 \hline
 $P_{12}(m)$ & $m^3 - a_{2,1} m^2 + a_{2,2} n m - a_{2,3}n^2$\\
 \hline
 $P_{14}(m)$ & $m^4 - m^3 - a_{3,1}nm^2 + a_{3,2}n^2m - a_{3,3}n^3$  \\
 \hline
 $P_{16}(m)$ & $m^4 + a_{4,1} m^3 -a_{4,2} n m^2 + a_{4,3} n^2 m - a_{4,4} n^3 $ \\
 \hline
\end{tabular}
\end{center}
\subsubsection{Step $3$}
For all polynomials we have $P_{\mathcal{G}}(0) \leq 0$, while they are monic which means that as $n \rightarrow \infty$, or consequently, $m \rightarrow \infty$, there exists $m_0 \geq 0$ such that for every $m \geq m_0$, $P_{\mathcal{G}}(m) \geq 0$. 
Thus, by Intermediate Value Theorem, polynomials have a positive real-valued root. Denote such a root as $x_\cG$ where $P_\cG(x_\cG) = 0$. 
Similarly, for $P_{14}(m)$ and $P_{16}(m)$, it can be shown that each also has a negative real-valued root. 
Let us denote the discriminants of the polynomials by $\cD$. As $n \to \infty$, it can be proved that $\cD < 0$ for sufficiently large $n$, as $P_{10}(0), P_{12}(0), P_{14}(0), P_{16}(0) < 0$. 
Since $\cD <0$, is it derived that there is only one real root for $P_{10}(m),$ and $P_{14}(m)$, and exactly two real-valued roots for $P_{14}(m)$, and $P_{16}(m)$. 
By solving the cubic and quartic equations, it is derived that:
\begin{itemize}
    \item $x_{10} = \cO(n^{(2/3)})$, $x_{12} = \cO(n^{(2/3)})$
    \item $x_{14} = \cO(n^{(3/4)})$, $x_{16} = \cO(n^{(3/4)})$
\end{itemize}

Thus, for each girth we have

\begin{equation}{\text{Girth $10$: }}
\label{bnd-g10}
    m \geq x_{10} = \cO(n^{(2/3)})
\end{equation}
\begin{equation}{\text{Girth $12$: }}
\label{bnd-g12}
    m \geq x_{12} = \cO(n^{(2/3)})
\end{equation}
\begin{equation}{\text{Girth $14$: }}
\label{bnd-g14}
    m \geq x_{14} = \cO(n^{(3/4)})
\end{equation}
\begin{equation}{\text{Girth $16$: }}
\label{bnd-g16}
    m \geq x_{16} = \cO(n^{(3/4)})
\end{equation}

\noindent{Remark $2$:}
Note that inequalities \eqref{thm-8bnd}, \eqref{bnd-g10}, \eqref{bnd-g12}, \eqref{bnd-g14}, and \eqref{bnd-g16} together state that it is not possible to have a small number of vertices, and high girth simultaneously. In fact, as the girth increases, the minimum number of check nodes required for that girth grows sublinearly but polynomially with $n$.

\section{Girth-Eight Regular Bipartite Graph Construction}
\label{sec:reg}
In this section, we present a construction of $(w_c, w_r)$-regular bipartite graphs for arbitrary integers $w_r \geq w_c \geq 2$, with girth $8$ and minimum distance $d_{\min} = 2^{w_c}$. The construction is intended for high-rate applications with long block lengths.

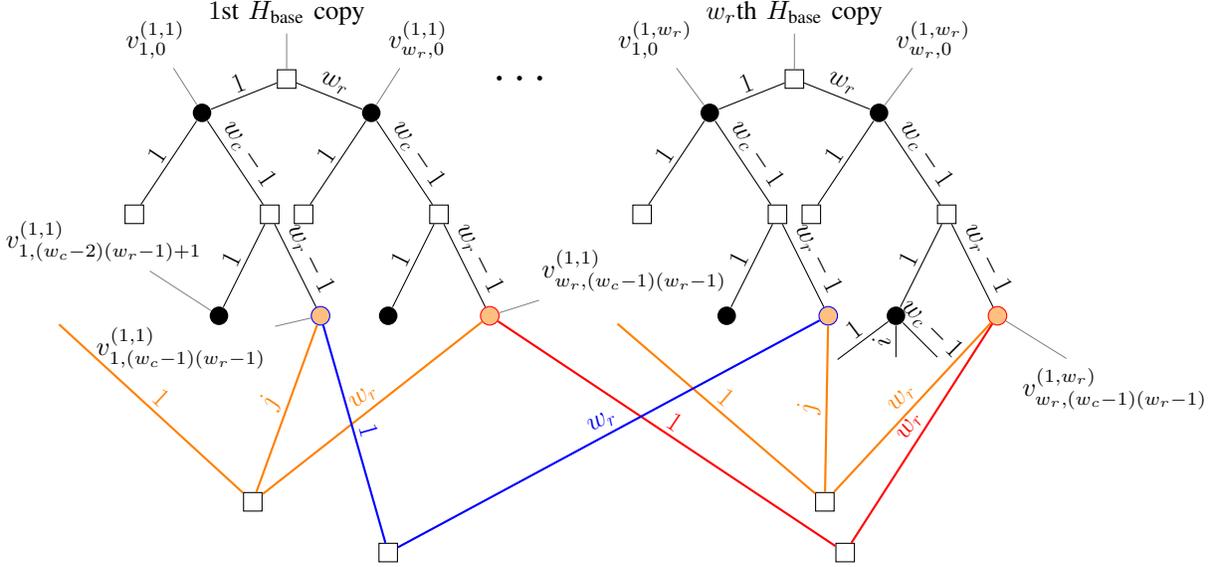
\begin{figure*}[ht]
    \centering
        \label{copytree}

\begin{tikzpicture}[
    scale=.45, 
    every node/.style={font=\normalsize}, 
    circ/.style={circle, fill=black, inner sep=0pt, minimum size=7pt},
    sq/.style={rectangle, draw, inner sep=0pt, minimum size=7pt}
]
\node[sq,pin=90:{$1$st $H_{\text{base}}$ copy}] (n11) at (0,0) {};

\node[circ,pin=100:{$v_{1,0}^{(1,1)}$}] (n21) at (-2.5,-1) {};
\node[circ,pin=80:{$v_{w_r,0}^{(1,1)}$}] (n22) at (2.5,-1) {};

\node[sq] (n31) at (-4.5,-4) {};
\node[sq] (n32) at (-0.5,-4) {};
\node[sq] (n33) at (0.5,-4) {};
\node[sq] (n34) at (4.5,-4) {}; 

\node[circ, draw = blue, fill=orange!50 ,pin=183:{\rotatebox{0}{$v_{1,(w_c-1)(w_r-1)}^{(1,1)}$}}] (n41) at (1,-7) {};
\node[circ,pin=100:{$v_{1,(w_c-2)(w_r-1)+1}^{(1,1)}$}] (n42) at (-2,-7) {};
\node[circ] (n43) at (3,-7) {};
\node[circ, draw = red, fill=orange!50,pin=20:{$v_{w_r,(w_c-1)(w_r-1)}^{(1,1)}$}] (n44) at (6,-7) {};

\node[] (n51) at (2,-8.5) {};
\node[] (n52) at (3.3,-8.5) {};
\node[] (n53) at (5.5,-8.5) {};

\node[] (n71) at (0,-8) {};
\node[] (n72) at (6,-8) {};
\node[] (n73) at (-7,-7) {};

\node[sq] (n61) at (-1,-12.5) {};

\begin{scope}[every node/.style={above,sloped,font=\LARGE,scale=.6}]
    \draw (n11) -- node {$1$} (n21);
    \draw (n11) -- node {$w_r$} (n22);

    \draw (n21) -- node {$1$} (n31);
    \draw (n21) -- node {$w_c-1$} (n32);
    \draw (n22) -- node {$1$} (n33);
    \draw (n22) -- node {$w_c-1$} (n34);

    \draw (n32) -- node {$w_r-1$} (n41);
    \draw (n32) -- node {$1$} (n42);
    \draw (n34) -- node {$1$} (n43);
    \draw (n34) -- node {$w_r-1$} (n44);


    \draw[orange, thick] (n61) -- node {$j$} (n41);
    \draw[orange, thick](n61) -- node {$w_r$} (n44);
    \draw[orange, thick] (n61) -- node {$1$} (n73);
\end{scope}

\node at (7,0) {\LARGE $\cdots$};

\node[sq,pin=90:{$w_r$th $H_{\text{base}}$ copy}] (n11b) at (15,0) {};

\node[circ,pin=100:{$v_{1,0}^{(1,w_r)}$}] (n21b) at (12.5,-1) {};
\node[circ,pin=80:{$v_{w_r,0}^{(1,w_r)}$}] (n22b) at (17.5,-1) {};

\node[sq] (n31b) at (10.5,-4) {};
\node[sq] (n32b) at (14.5,-4) {};
\node[sq] (n33b) at (15.5,-4) {};
\node[sq] (n34b) at (19.5,-4) {};

\node[circ, draw = blue, fill=orange!50] (n41b) at (16,-7) {};
\node[circ] (n42b) at (13,-7) {};
\node[circ] (n43b) at (18,-7) {};
\node[circ, draw = red, fill=orange!50,pin=290:{\rotatebox{0}{$v_{w_r,(w_c-1)(w_r-1)}^{(1,w_r)}$}}] (n44b) at (21,-7) {};

\node[] (n51b) at (16,-8.5) {};
\node[] (n52b) at (18,-8.5) {};
\node[] (n53b) at (19.5,-8.5) {};

\node[] (n71b) at (15,-8) {};
\node[] (n72b) at (21,-8) {};
\node[] (n73b) at (9.5,-7) {};

\node[sq] (n61b) at (15.9,-12.5) {};

\begin{scope}[every node/.style={above,sloped,font=\LARGE,scale=.6}]
    \draw (n11b) -- node {$1$} (n21b);
    \draw (n11b) -- node {$w_r$} (n22b);

    \draw (n21b) -- node {$1$} (n31b);
    \draw (n21b) -- node {$w_c-1$} (n32b);
    \draw (n22b) -- node {$1$} (n33b);
    \draw (n22b) -- node {$w_c-1$} (n34b);

    \draw (n32b) -- node {$w_r-1$} (n41b);
    \draw (n32b) -- node {$1$} (n42b);
    \draw (n34b) -- node {$1$} (n43b);
    \draw (n34b) -- node {$w_r-1$} (n44b);

    \draw (n43b) -- node {$1$} (n51b);
    \draw (n43b) -- node {$i$} (n52b);
    \draw (n43b) -- node {$w_c-1$} (n53b);

    \draw[orange, thick] (n61b) -- node {$j$} (n41b);
    \draw[orange, thick] (n61b) -- node {$w_r$} (n44b);
    \draw[orange, thick] (n61b) -- node {$1$} (n73b);
\end{scope}

\node[sq] (n81b) at (16.5,-14) {};
\node[sq] (n82b) at (3,-14) {};
\begin{scope}[every node/.style={above,sloped,font=\LARGE,scale=.6}]
\draw[red, thick] (n81b) -- node {$w_r$} (n44b);
\draw[red, thick] (n81b) -- node {$1$} (n44);
\draw[blue, thick] (n82b) -- node {$w_r$} (n41b);
\draw[blue, thick] (n82b) -- node {$1$} (n41);
\end{scope}

\end{tikzpicture}
\label{fig:tiki}
\caption{$(w_r, w_c)$-regular $H_\text{reg}$ after the first iteration of \text{Level $3$}}
\end{figure*}

We define a base graph $H_\text{base}$ as a three-layer tree rooted at a check node, in which each check node has degree $w_r$, and each variable node has degree $w_c$, except those located in the final layer. The structure is illustrated in Fig. $1$. To obtain the final layer variable nodes with degree $w_c$, we employ a recursive process that introduces additional sets of check nodes, progressively extending the structure. This construction is described in Algorithm 1.

{\small
\begin{algorithm}
\label{reg-def}
\caption{\((w_c, w_r)\)-regular graph \(H_{\text{reg}}\) with girth $8$}
\begin{algorithmic}[1]

\State \textbf{Input:} Column weight $w_c$, row weight $w_r$
\State \textbf{Output:} \(H_{\text{reg}}\)
\State \textbf{Level 1: Base Graph Construction}
\State Construct base tree $H_{\text{base}}$ with parameters $w_r$ and $w_c$.
\State Define first layer variable nodes as $v_{1,0}, v_{2,0}, \dots, v_{w_r,0}$.
\For{$i = 1$ to $w_r$}
    \State Denote second-layer variable nodes rooted at $v_{i,0}$ by:
    \[
        \{v_{i,1}, \dots, v_{i,(w_c - 1)(w_r - 1)}\}
    \]
\EndFor
\State \textbf{Level 2: Adding Extra Check Nodes}
\For{$j = 1$ to $(w_c - 1)(w_r - 1)$}
    \State Add check node $c_j$ connecting $\{v_{1,j}, \dots, v_{w_r,j}\}$
\EndFor

\If{$w_c \geq 3$}
    \State \textbf{Level 3: Recursive Expansion}
    \For{$s = 1$ to $w_c - 2$}
        \State Create $w_r$ copies of the current graph
        \State Label variable nodes in the $t$-th copy as $v^{(s,t)}_{i,j}$
        \For{$j = 1$ to $w_r^s(w_c-1)(w_r-1)$}
            \State Add check node $c^{(s)}_j$ connecting
            \[
                \{v^{(s,1)}_{i,j}, v^{(s,2)}_{i,j}, \dots, v^{(s,w_r)}_{i,j}\}
            \]
        \EndFor
    \EndFor
\EndIf

\end{algorithmic}
\end{algorithm}
}

\begin{proposition}
\label{reg-dim}
    For given $w_r \geq w_c \geq 2$, the corresponding $H_\text{reg}$ constructed by Algorithm 1, has dimensions
    \begin{equation}
        n = w_r^{(w_c-2)}(w_r + w_r(w_r-1)(w_c - 1)),
    \end{equation}
    \begin{equation}
        m = w_r^{(w_c-2)}(1 + w_r(w_c - 1) + (w_r - 1)(w_c - 1)^2).
    \end{equation}
\end{proposition}
\begin{proof}
    The proof is derived by counting the total number of variable nodes and the required number of copies.
\end{proof}
\noindent{\text{Remark $3$:}}
By Proposition \ref{reg-dim}, for $w_c = 3$, we have $n = \theta(w_r^3)$ and $m = \theta(w_r^2)$. In \cite{furedi1995graphs}, Furedi et al. proposed constructions with arbitrarily large girths with notably small sizes. However, the size estimates for both of these constructions are of the order $\frac{\cG}{2}$, with the exponent being approximately $4$ for $\cG = 8$. In contrast, the exponent for $H_\text{reg}$ with respect to $w_r$ is smaller than $\frac{\cG}{2} = 4$, while still maintaining a construction with lower complexity.

There exists a unique array that characterizes each variable node based on the indices it reaches during the construction in Algorithm 1. For any variable node $v$, we define its characteristic array as
\begin{equation}
    a_v = (t_1, \dots, t_{w_c-2}, \ell, j)
\end{equation}
where the indices $0 \leq t_1, \dots, t_{w_c-2} \leq w_r$ correspond to the copies where $H_{\text{base}}$ of $v$ is placed, and the indices $\ell$ and $j$ represent the first two indices that the variable nodes in $H_{\text{base}}$ receive at \text{Level $1$}.

\begin{proposition}
\label{reg-prop-gdmin}
    The girth of $H_{\text{reg}}$ is $\cG = 8$, and the minimum distance is  $d_{\text{min}} = 2^{w_c}$.
\end{proposition}
Before proving Proposition \ref{reg-prop-gdmin}, we first state two observations about $H_\text{reg}$.

\noindent{\text{Observation $1$:}}\label{reg-obs1}
Consider two vertices $v_1$ and $v_2$ with characteristic arrays $a_1 = (t_{1,1},...,t_{1,w_c-2}, \ell_1, j_1)$ and $a_2 = (t_{2,1},...,t_{2,w_c-2}, \ell_2, j_2)$. If a check node connects $v_1$ and $v_2$, then $a_1$ and $a_2$ are identical at $w_c-2$ indices and differ in exactly one index.

\noindent{\text{Observation 2:}}\label{reg-obs2}
Assume that check node $c_{12}$ connects two vertices $v_1$ and $v_2$, and $c_{13}$ connects $v_1$ and $v_3$, with characteristic arrays $a_1, a_2$, and $a_3$ respectively. If $a_1$ and $a_2$ differ only at the $i$-th index, and $a_1$ and $a_3$ also differ only at the $i$-th index, then we must have $c_{12} = c_{13}$. 

\begin{proof}
    Let us denote each vertex $v$ by its characteristic array $a$. First, note that there is no pair of variable nodes $(v_1, v_2)$ with index arrays $a_1, a_2$ that two check nodes connecting them. Otherwise, by Observation $1$, since $v_1$ and $v_2$ are connected to each other by check nodes then, $a_1$ and $a_2$ are different in exactly one index. Thus, by Observation $2$, two check nodes are not distinct. This proves that $\cG \geq 6$.

    Now, suppose there exists a $6$-cycle consists of variable nodes $v_1, v_2$, and $v_3$ in $H_\text{reg}$, each pair connected to a distinct check node. Then, by Observation $1$, each pair of $a_1, a_2, a_3$ have exactly one different index with each other. However, by Observation $2$, these indices should be distinct which results in two of $a_i$ and $a_j$ be different in two different indices which is contradiction. Thus, $H_{\text{reg}}$ has no $6$-cycle.

    Now, we prove $d_{min} = 2^{w_c}$. For simplicity, first we prove the case $w_c = 3$, and the general case $w_c \geq 4$ proof is straightforward using induction. Suppose there are columns $\cV = \{v_1,..., v_r\}$ which have binary-sum equal to zero. Then, it can be seen that at least one of the $v_i$'s should be from the last layer variable nodes in $H_{\text{base}}$. Otherwise, all of the check nodes connecting $v_{i,0}$ have degree $1$ which do not sum to zero. Without loss of generality suppose it is the vertex $v^{(1,1)}_{1,1}$. All of the check nodes connected to $v^{(1,1)}_{1,1}$ should have at least two neighbors in the set $\cV$. Thus, there exist $v^{(1,1)}_{1,j}, v^{(1,1)}_{\ell,1}, v^{(1,t)}_{1,1} \in \cV$, for $\ell, t, j \neq 1$. Again, by similar logic since $v^{(1,1)}_{\ell,1} \in \cV$, there should also exist $v^{(1,1)}_{\ell,j_1} \in \cV$, for $j_1 \neq 1$. Thus, for base matrix $H_\text{base}^{(1,1)}$, we have at least $4$ variable nodes inside $\cV$. Moreover, since $v^{(1,t)}_{1,1} \in \cV$, by similar logic, $H_{\text{base}}^{(1,t)}$ has also at least four vertices in $\cV$. Therefore, $|\cV| \geq 8$. In fact the set of columns $\cV = \{v^{(1,1)}_{1,1}, v^{(1,1)}_{1,2}, v^{(1,1)}_{2,1}, v^{(1,1)}_{2,2},v^{(1,2)}_{1,1}, v^{(1,2)}_{1,2}, v^{(1,2)}_{2,1}, v^{(1,2)}_{2,2}\}$ has binary sum of zero. Thus, $d_{min} = 8$.
\end{proof}

\noindent{\text{Remark $4$:}}
As shown in \cite{hu2001progressive} and \cite{tanner1981recursive}, for LDPC codes with column weight $w_c$, the general minimum distance bound for $\cG = 8$ is $d_{min} \geq 2w_c$. However, for the regular LDPC code considered here, as derived in Proposition \ref{reg-prop-gdmin}, the minimum distance is $d_{min} = 2^{w_c}$, which grows exponentially with $w_c$, providing a significant improvement over the linear bound.

\section{Semi-Regular Construction}
\label{sec:semi}
In this section, we construct a $3$-semi-regular bipartite graph using free $3$-arithmetic progression sequences.
\begin{definition}[Free $3$-Arithmetic Progression ($3$-AP)]
A sequence of integers \( \{a_1, a_2, a_3, \dots \} \) is called a \textit{free $3$-arithmetic progression} ($3$-AP) if no three distinct elements \( a_i, a_j, a_k \) form an arithmetic progression, or equivalently, if
\(
a_i + a_j = 2a_k\) does not hold for any distinct \(i, j, k.
\)
\end{definition}
Although finding such sequences may seem challenging, several studies have shown that these sequences exhibit an approximately linear size. Furthermore, empirical results provide nearly optimal examples of such sequences up to some $N$, where $N$ represents the length of the sequence \cite{behrend1946sets,dybizbanski2012sequences,gasarch2008finding,salem1942sets}.
\begin{definition}
    Let $r(M)$ be the maximum size of a free $3$-AP sequence in $[M] \triangleq \{1,2,\dots,M\}$.
\end{definition}
\begin{theorem}[\cite{behrend1946sets}, \cite{gasarch2008finding}]
\label{lem-3aplin}
    There exists constants $c_1$ and $c_2$ such that 
    \begin{equation}
        c_1M^{1-c_2/\sqrt{\log M}} \leq r(M).
    \end{equation}
\end{theorem}

Consider an increasing free $3$-AP sequence $b_i$. Using this sequence, we construct the parity-check matrix $H_s$, where each column has a weight of $3$, and the matrix has a girth of $8$.

\begin{definition}
    Let $n = t^2$ for some integer $t \geq 0$. The matrix $H_s$ of dimension $m \times n$ has three distinct levels of rows:
\begin{itemize}
    \item \text{Level $1$}: Contains the first $1$s of all columns (rows $1$ to $t$).
    \item \text{Level $2$}: Contains the second $1$s of all columns (rows $t+1$ to $2t+b_t-1$).
    \item \text{Level $3$}: Contains the third $1$s of all columns (rows $2t+b_t$ to $4t + 2b_t -3$).
\end{itemize}
For each column $j$ (where $ j = 1, 2, \dots, n$), the row positions of the three $1$s are determined as follows:

\subsubsection*{First Level ($ r_{1j} $)}
The first $1$ in column $j$ is placed in row $a_j$, chosen arbitrarily from Level $1$, such that each row has an equal weight of $t$:
\begin{equation}
\label{eq-lvl1}
    r_{1j} = a_j \quad \text{where} \quad a_j \in \text{Level $1$}.
\end{equation} 
\subsubsection*{Second Level ($r_{2j}$)}
The second $1$ in column $j$ is placed at row $r_{2j}$, given by:
\begin{equation}
\label{eq-lvl2}
    r_{2j} = c_1 + b_{i_j} + r_{1j}
\end{equation}
where $c_1 = \sqrt{n} - 1$, and $i_j$ is chosen such that for any $\ell \neq k$ where $a_\ell = a_k$, we have $i_\ell \neq i_k$.

\subsubsection*{Third Level ($r_{3j}$)}
The third $1$ in column $j$ is placed at row $r_{3j}$, given by:
\begin{equation}
\label{eq-lvl3}
r_{3j} = c_2 + a_j + r_{2j}
\end{equation}
where $c_2 = \sqrt{n} - 1 + b_t$.
\end{definition}

Note that the constants in equations \eqref{eq-lvl2} and \eqref{eq-lvl3} are chosen to ensure that the $2$nd- and $3$rd-level $1$s are placed in distinct levels.

\begin{figure}[t]\vspace{-0.05in}
	\centering
	\includegraphics[scale=0.35]{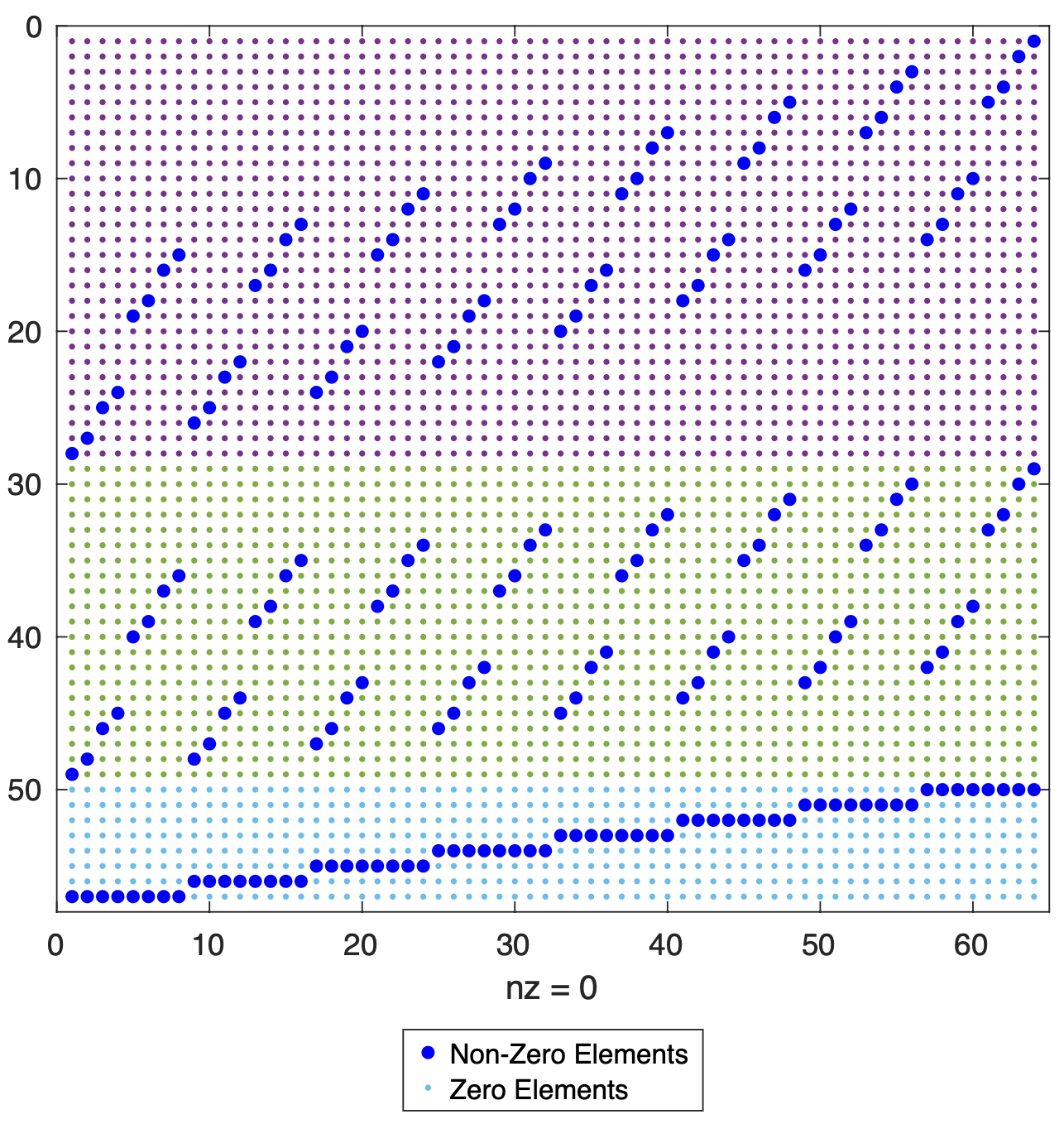}
	\caption{$n=64$, $\cG = 8$, $b = \{1, 2, 4, 5, 10, 11, 13, 14\}$.}
    \label{3semi-ex}
	\vspace{-0.2in}
\end{figure}
\noindent{\text{Example $2$:}} Fig. \ref{3semi-ex}. shows an example of a $3$-semi-regular matrix with dimensions $n = 64,\ m = 57$ derived by the sequence b from \cite{dybizbanski2012sequences}.
\begin{proposition}
    The tanner graph $\cG_s$ corresponding to the parity-check matrix $H_s$ has girth $8$, and the code $\cC_s$ with parity-check $H_s$ has minimum distance of at least $6$.
\end{proposition}
\begin{proof}
First note that due to the conditions \eqref{eq-lvl1}, \eqref{eq-lvl2}, \eqref{eq-lvl3} we can not have two columns with more than one intersection, and consequently, a $4$-cycle can not happen.
    Now, suppose there exist three columns such that each pair of them has exactly one distinct intersection. Note that due to the level based construction, we can not have two intersections in two rows at the same level. Without loss of generality suppose there are three columns $p, q, \ell$ such that $r_{1p} = r_{1q}$, $r_{2q} = r_{2\ell}$, and $r_{3p} = r_{3\ell}$. Thus, we have 
    \vspace{-0.5mm}
    \begin{equation*}
        r_{3p} - r_{1p} = r_{2q} - r_{1q} + r_{3\ell} - r_{2\ell},
    \end{equation*}
    or equivalently,
    \begin{equation}
    \label{eq-cond1}
        c_2 + c_1 + a_p + b_{i_p} = c_1 + b_{i_q} + c_2 + a_{\ell},
    \end{equation}
    \begin{equation}
        \label{eq-cond2}
        a_p = a_q,
    \end{equation} 
    \begin{equation}   
    \label{eq-cond3}
        a_q + c_1 + b_{i_q} = a_\ell + c_1 + b_{i_\ell}.
    \end{equation}
    Then, equations \eqref{eq-cond1}, \eqref{eq-cond2}, and \eqref{eq-cond3} together, concludes
    \begin{equation}
        b_p + b_\ell = 2b_q,
    \end{equation}
    which is contradiction.
    Since the girth of $\cG_s$ is $8$, by \cite[Theorem 2]{hu2001progressive}, we have 
    \vspace{-3mm}
    \begin{equation}
        d_{min} \geq 6.
    \end{equation}
\end{proof}
\vspace{-0.5mm}
\noindent{\text{Remark $5$:}}
As $n = t^2$ grows larger, again by Theorem \ref{lem-3aplin}, $r(M)$ approaches a nearly linear function of $M$, implying $M = \cO(t)$ as $n \to \infty$. Since $H_s$ has $m = 4t + 2b_t - 3$ rows, it follows that $m = \cO(t)$ as $n \to \infty$. Consequently, the rate of the code $\cC_s$ is 
\[
R = 1 - \frac{m}{n} \approx 1 - \frac{1}{t} \to 1 \quad \text{as } n \to \infty.
\]

\section{Conclusion and Future Directions}
\label{sec:conclusion}

In this paper, we studied the relations between the girth of a bipartite graph and the number of variable and check nodes. We derived bounds on the size of the vertices in regular bipartite graphs given constraints on the girth as well as on the column weight. Furthermore, we presented two constructions for the bipartite graph constructions with girth $\cG = 8$. 

There are several directions for future work. As we discussed in Remark $1$, the lower bound presented in Theorem \ref{thm-8bnd} is tight at least for $w_c=2$. However, our results leave the tightness of this bound for general $w_c$ as an open problem, which is an interesting direction for future research. In other words, either constructions exist to achieve this bound for any $w_c$, or the lower bound can be improved. Furthermore, when bipartite graphs are used as the Tanner graph of a code, the interaction between the minimum distance of the resulting code and the girth of the Tanner graph, and how the lower bound of the required redundancy $m$ could be revised to take both parameters into account is another direction for future research.

\bibliographystyle{IEEEtran}
{\footnotesize \bibliography{reffff}}

\end{document}